%% file: main.tex
\documentclass[10pt,conference, letterpaper]{IEEEtran}
\IEEEoverridecommandlockouts

\usepackage{configuration}
\usepackage{booktabs}
\usepackage{amsthm}
\usepackage{amsmath}
\usepackage{amssymb}
\usepackage{booktabs}
\usepackage{enumitem}

\usepackage{lipsum}

\newtheorem{corollary}{Corollary}
\newtheorem{proposition}{Proposition}
\newtheorem{definition}{Definition}

\input{acro.cfg}

\begin{document}
\bstctlcite{IEEEexample:BSTcontrol}


\title{Understanding Conflict and Compatibility Conditions for Agents in AI-RAN}

\author{
  \IEEEauthorblockN{
    Arshia Zolghadr\IEEEauthorrefmark{2},
    Joao F. Santos\IEEEauthorrefmark{2},
    Imtiaz Nasim\IEEEauthorrefmark{3},
    Deniz Aytemiz\IEEEauthorrefmark{3},
    Nicholas J. Kaminski\IEEEauthorrefmark{3}, and
    Jacek Kibi\l{}da\IEEEauthorrefmark{2}
  }
  \IEEEauthorblockA{
    \IEEEauthorrefmark{2}\textit{Commonwealth Cyber Initiative},
    \textit{Virginia Tech}, USA,
    e-mail: \{arshiaz,joaosantos,jkibilda\}@vt.edu}
\IEEEauthorblockA{
    \IEEEauthorrefmark{3}\textit{Idaho National Laboratory}, USA,
    e-mail: \{imtiaz.nasim,deniz.aytemiz,nicholas.kaminski\}@inl.edu}
}
\maketitle

\begin{abstract}
Disaggregated mobile networks expose open interfaces for independent agents to control the \ac{RAN}. The interactions between agents with contrasting objectives can result in conflicts. However, existing conflict mitigation methods treat agents that control the same parameter as being in a conflict that requires arbitration of their actions, without assessing whether their objectives can be satisfied simultaneously. In this paper, we introduce a control-theoretic approach to determine when agents can coexist and share control of the same parameters. We model threshold-based agents often encountered in the literature as sampled deadband feedback controllers and represent their objectives through \ac{KPI} targets and tolerances that define acceptable deadband bounds. We derive compatibility conditions from these bounds and show that compatibility for any number of agents and monitored \acp{KPI} can be established through pairwise tests. We propose a conflict mitigation policy that allows compatible agents to share control of the same network parameters while restricting arbitration to cases where their objectives cannot be satisfied simultaneously. We evaluate our policy through simulations using NIST's \texttt{ns3-oran} extension of \texttt{ns-3}, and our results show that it enables agents to share control with near-zero \ac{SLA} drift when their objectives are compatible and the shared acceptable range is attainable, while matching Priority Mitigation when their objectives are incompatible.
\end{abstract}



\acresetall

\fancypagestyle{firstpage}
{
    \fancyhead[L]{This work has been submitted to the IEEE for possible
      publication.\\
      Copyright may be transferred without notice, after which this version may no longer be accessible.}
    \fancyhead[R]{}
}
\thispagestyle{firstpage}

\input{sections/01_introduction.tex}
\input{sections/02_related_works}
\input{sections/04_methods}
\input{sections/06_evaluation.tex}
\input{sections/07_conclusion.tex}

\section*{Acknowledgements}

Work supported through the INL Laboratory Directed Research \& Development
(LDRD) Program under DOE Idaho Operations Office Contract DE-AC07-05ID14517.
This research was also partially funded by the NSF IUCRC WISPER, under NSF Award
No. 2412872 and by the NSF US-Ireland R\&D Partnership program under NSF Award No. 2421362.

\bibliographystyle{IEEEtran}
\bibliography{bibliography}

\end{document}

%% file: sections/01_introduction.tex
\section{Introduction}\label{sec:intro}

Mobile networks are increasingly adopting AI-driven optimization of \ac{RAN} to achieve different objectives. This evolution is prominently reflected in the AI-RAN vision, which has gained substantial momentum through the establishment of a global consortium with broad participation from industry stakeholders~\cite{airan_alliance_whitepaper}. Such AI-driven optimization can be realized through a collection of third-party agents deployed on the network infrastructure, e.g., dApps, xApps, and rApps in the context of \ac{O-RAN}~\cite{santos2025managing}, that monitor \acp{KPI} and make decisions to modify \ac{RAN} parameters. However, agents can make conflicting decisions with the potential to degrade performance, cause instability, or disrupt network operation~\cite{del2025pacifista}. 

The existing literature on conflicts in mobile networks focuses on conflict detection and conflict mitigation as distinct problems. Conflict detection identifies interactions among agents, \ac{RAN} parameters, and \acp{KPI} either based on agent subscription information or RAN data analysis~\cite{del2026predicting,del2025pacifista,zolghadr2025learning,santos2026conflict}. Mitigation methods determine which actions take effect through different resolution mechanisms~\cite{adamczyk2023conflict,wadud2024qacm,giannopoulos2025comix}, assuming known or detected conflicts. While effective, this approach fails to recognize that not all detectable conflicts necessitate mitigation and that considering agents' objectives can improve performance. For instance, in the case of a direct conflict, defined by \ac{O-RAN} as two or more agents controlling the same \ac{RAN} parameter~\cite{oran_conflict_mitigation_2024}, existing works typically treat any detected conflict as warranting arbitration. However, consider a throughput-maximization agent and an energy-saving agent that both monitor downlink throughput and control the \ac{PRB} allocation~\cite{del2025pacifista,giannopoulos2025comix}. The throughput-maximization agent requests more \acp{PRB} when the measured throughput falls below its target, while the energy-saving agent requests fewer \acp{PRB} when the throughput is above its target. Although both control the same parameter and are in direct conflict over it, if the throughput performance stays within both of their targets, this can satisfy both objectives, making arbitration unnecessary. Motivated by this observation, we focus on designing a mitigation policy that accounts for agents' objectives, enabling the agents to share control over parameters.

To achieve this goal, we model a \ac{RAN} agent as a sampled deadband feedback controller model~\cite{khalil2002nonlinear, gelb1968multiple}, reflecting the behavior of  a typical \ac{RAN} agent that monitors a given \ac{KPI} and only takes action when its value falls outside its target objective~\cite{del2025pacifista}.  
Then, we show that a set of agents are \textit{compatible} if there exists at least one network state, as characterized by their monitored \acp{KPI}, in which the objectives of all agents are simultaneously satisfied.
On the contrary, a collection of agents is \emph{incompatible} when no such network state exists, so at least one objective must be violated. This distinction allows us to propose 
a conflict mitigation policy that arbitrates only between incompatible \ac{RAN} agents, thereby improving network performance and target performance compliance. 

The main contributions of this paper are as follows:
\begin{itemize}[leftmargin=1em]
    \item We provide the first control-theoretic formulation of AI-RAN agents and define their compatibility conditions. 
    \item 
    We show that the compatibility of an arbitrary number of agents over an arbitrary number of \acp{KPI} can be established through pairwise tests.
    \item We propose a compatibility-aware conflict mitigation policy that enables compatible agents to share control of a parameter, while resorting to prioritizing a single agent only when their objectives are incompatible.
    \item We quantify the impact of conflicts and compare state-of-the-art priority conflict mitigation policy~\cite{adamczyk2023conflict} against our proposed method in a representative AI-RAN scenario using NIST's \texttt{ns3-oran} extension of the \texttt{ns-3} simulator~\cite{garey2023oran}. 
\end{itemize}

%% file: sections/02_related_works.tex
\section{Related Works}\label{sec:related}
There is a growing literature exploring methods to
detect conflicts based on the interactions between agents, \ac{RAN} parameters, and \acp{KPI}.
For example, the work of~\cite{del2025pacifista} 
adopted statistical profiles of agents in a sandbox environment to infer interactions 
between agents, parameters, and \acp{KPI}. 
In our prior works, 
we proposed data-driven methods to identify interactions based on the correlation~\cite{zolghadr2025learning} and similarity~\cite{santos2026conflict} between parameters and \acp{KPI}.
While differing in their methods to identify interactions, existing works classify agents that control or influence the same parameter as being in conflict and require mitigation. This stringent approach overlooks cases in which multiple agents have compatible objectives and can share control or influence over a parameter without negatively affecting each other. In other words, existing conflict detection methods identify which agents interact on the same parameters, but do not determine whether mitigation is actually necessary.

Another group of works explores methods to mitigate conflicts after they are detected.
For example, distillation-based methods explore knowledge distillation to merge multiple \ac{DRL}-based agents into a single aggregate agent trained to prevent conflicts~\cite{erdol2026xapp}. However, this approach forfeits the modularity of independent agents, requires retraining, and can only be applied to \ac{DRL}-based agents. 
Alternatively, prioritization-based methods (also referred to as rule-based in some papers) rank agents according to an arbitrary priority and only enact the actions of the highest-ranked agent~\cite{adamczyk2023conflict, wadud2024qacm, wadud2025xapp, giannopoulos2025comix}.
While these works adopt different prioritization strategies, they rely on the same mechanism: the actions of one agent take precedence, and others' are suppressed, either permanently or at each decision epoch. This suppression is applied regardless of whether the other agents' actions would conflict with the prioritized agent or contribute to achieving the same (or a similar) objective.

The existing literature focuses on preventing potential effects arising from multiple agents interacting over the same parameters, without assessing the compatibility of their objectives. 
As a result, existing conflict mitigation methods can unnecessarily restrict the number of agents that can simultaneously control the network infrastructure, potentially curbing the achievable performance.
To address this limitation, we examine the behavior of agents operating on the network infrastructure and leverage control-theoretic concepts to derive the conditions under which multiple agents can coexist and share control of the same parameters.

%% file: sections/04_methods.tex
\section{System Model}\label{sec:methods}
In the following, we present the \ac{RAN} model and then define the agents as sampled deadband feedback controllers. 
We focus on the \ac{Near-RT RIC} because it is the \ac{RAN} control domain in which xApps use E2 measurements and actions for near-real-time control and optimization of RAN elements and resources, and RAN conflict mitigation is defined for multi-xApp operation~\cite{oran_wg3_ricarch,oran_conflict_mitigation_2024}. However, the framework can potentially be extended to other control domains of AI-RAN.

\subsection{RAN Model}\label{subsec:ran_model}

We consider a set of agents $\mathcal{A}=\{1,\ldots,A\}$ operating in the \ac{Near-RT RIC}.
Let $\mathcal{P}$ denote the set of \ac{RAN} parameters exposed to the agents, e.g., \ac{PRB} allocation, transmit power, and handover threshold, and let $\mathcal{K}$ denote the set of \acp{KPI} available from the \ac{RAN}.
Each agent $a\in\mathcal{A}$ controls a subset $\mathcal{P}_a\subseteq\mathcal{P}$ and monitors a subset $\mathcal{K}_a\subseteq\mathcal{K}$.

Following the \ac{O-RAN} definition of a direct conflict~\cite{oran_conflict_mitigation_2024}, for each $p\in\mathcal{P}$ we define the set of agents that control it as
\begin{equation}\label{eq:parameter_group}
    \mathcal{A}_p=\{a\in\mathcal{A}:p\in\mathcal{P}_a\}.
\end{equation}
Thus, there is a conflict over parameter $p$ when $|\mathcal{A}_p|\geq2$. For such a parameter, let
\begin{equation}\label{eq:parameter_kpi_set}
    \mathcal{K}_p=\bigcup_{a\in\mathcal{A}_p}\mathcal{K}_a
\end{equation}
denote the \acp{KPI} monitored by at least one of the agents controlling $p$.

For each $k\in\mathcal{K}_p$, let $y_k(t)$ denote the reported value of \ac{KPI} $k$ at time $t$. We collect the \acp{KPI} relevant to parameter $p$ in the vector
\begin{equation}\label{eq:kpi_vector}
    \mathbf{y}(t)=\bigl(y_k(t)\bigr)_{k\in\mathcal{K}_p}.
\end{equation}

\subsection{Modeling the Operation of RAN Agents}\label{subsec:threshold}

Agents monitor \acp{KPI} and act only when a monitored \ac{KPI} leaves a range of values they consider satisfactory.
For example, a throughput maximization agent only requests additional resources when the throughput falls below its intended threshold~\cite{del2025pacifista}.
This behavior corresponds to a sampled deadband feedback controller:
the agent remains inactive while its monitored \ac{KPI} lies within a range and acts to steer it back when it leaves that range~\cite{khalil2002nonlinear,gelb1968multiple}. Leveraging the language of control theory, we define each RAN agent through the following characteristics.

For every \ac{KPI} $k\in\mathcal{K}_a$ monitored by agent $a$, let $r_{a,k}$ denote the agent's target for $k$ and $\delta_{a,k}>0$ the agent's tolerance. At a time instant $t$, the control error is then
\begin{equation}\label{eq:deadband_error}
    e_{a,k}(t)=y_k(t)-r_{a,k}.
\end{equation}
We say that the \ac{KPI} $k$ is within the agent's deadband when
\begin{equation}\label{eq:deadband_condition}
    |e_{a,k}(t)|\leq\delta_{a,k},
\end{equation}
with lower and upper bounds
\begin{equation}\label{eq:deadband_interval}
    l_{a,k}=r_{a,k}-\delta_{a,k},\qquad
    h_{a,k}=r_{a,k}+\delta_{a,k}.
\end{equation}

\section{Compatibility Analysis}\label{sec:conditions}

In this section, we define objective compatibility for the agents controlling a common \ac{RAN} parameter and establish how this compatibility extends from pairs of agents to the full set.

Objective compatibility concerns whether the agents' deadbands can be satisfied simultaneously. Let
\begin{equation}\label{eq:kpi_agents}
    \mathcal{A}_{p,k}
    =
    \{a\in\mathcal{A}_p : k\in\mathcal{K}_a\}
\end{equation}
denote the set of agents with a direct conflict over parameter $p$ that monitor \ac{KPI} $k$. For each $k\in\mathcal{K}_p$, define
\begin{equation}\label{eq:aggregate_bounds}
    L_k
    =
    \max_{a\in\mathcal{A}_{p,k}} l_{a,k},
    \qquad
    H_k
    =
    \min_{a\in\mathcal{A}_{p,k}} h_{a,k},
\end{equation}
as the maximum lower bound and minimum upper bound on \ac{KPI} $k$ across the conflicting agents' deadbands.

\begin{definition}[Objective Compatibility]\label{def:compatibility}
The agents in $\mathcal{A}_p$ are said to be \emph{compatible} if
\begin{equation}\label{eq:compat_condition}
    L_k\leq H_k,
    \qquad
    \forall k\in\mathcal{K}_p.
\end{equation}
For each $k \in \mathcal{K}_p$, the interval $[L_k,H_k]$ is called the \emph{shared acceptable range} for \ac{KPI} $k$. If~\eqref{eq:compat_condition} fails for any $k\in\mathcal{K}_p$, the agents are \emph{incompatible}.
\end{definition}

The shared acceptable range contains the KPI $k$ values that fall within the deadband of every agent in $\mathcal{A}_{p,k}$.
We next show that for a set of agents, it suffices to perform pairwise compatibility tests, which will be useful for implementing an efficient conflict mitigation mechanism.

\begin{proposition}[Pairwise Sufficiency]\label{prop:pairwise}
The agents in $\mathcal{A}_p$ are compatible if and only if, for every pair of distinct agents $i,j\in\mathcal{A}_p$ and every \ac{KPI} $k\in\mathcal{K}_i\cap\mathcal{K}_j$,
\begin{equation}\label{eq:pairwise_condition}
    \max(l_{i,k},l_{j,k})
    \leq
    \min(h_{i,k},h_{j,k}).
\end{equation}
\end{proposition}
There are $\binom{|\mathcal{A}_p|}{2}$ agent pairs to evaluate, with each pair $i,j$ requiring $|\mathcal{K}_i\cap\mathcal{K}_j|$ \ac{KPI} checks.

\begin{IEEEproof}
Suppose first that the agents in $\mathcal{A}_p$ are compatible. For any distinct $i,j\in\mathcal{A}_p$ and any $k\in\mathcal{K}_i\cap\mathcal{K}_j$, Definition~\ref{def:compatibility} gives $L_k\leq H_k$. Since $i,j\in\mathcal{A}_{p,k}$,
\begin{equation}   \max(l_{i,k},l_{j,k})
    \leq L_k
    \leq H_k
    \leq \min(h_{i,k},h_{j,k}),
\end{equation}
and therefore~\eqref{eq:pairwise_condition} holds.

Then, suppose~\eqref{eq:pairwise_condition} holds for every pair and fix any $k\in\mathcal{K}_p$. Since $\mathcal{A}_{p,k}$ is finite and non-empty, there exist $i,j\in\mathcal{A}_{p,k}$ such that
\begin{equation}
    L_k=l_{i,k},
    \qquad
    H_k=h_{j,k}.
\end{equation}

If $i=j$, then $L_k\leq H_k$ because $l_{i,k}\leq h_{i,k}$. If $i\neq j$,~\eqref{eq:pairwise_condition} implies $l_{i,k}\leq h_{j,k}$, and therefore $L_k\leq H_k$. Since this holds for every $k\in\mathcal{K}_p$, the agents in $\mathcal{A}_p$ are compatible by Definition~\ref{def:compatibility}.
\end{IEEEproof}

The above result is the one-dimensional case of Helly's theorem~\cite{helly1923mengen}. 
Applying this independently to each $k\in\mathcal{K}_p$ gives the multi-\ac{KPI} result.

\begin{corollary}[Two-Agent Compatibility]\label{cor:two_agent}
Consider two agents monitoring the same \ac{KPI}, with targets $r_1$ and $r_2$ and tolerances $\delta_1$ and $\delta_2$. Let
\begin{equation}
    \Delta r = |r_1-r_2|
\end{equation}
denote their target separation. The agents are compatible if and only if
\begin{equation}\label{eq:two_agent_condition}
    \Delta r \leq \delta_1+\delta_2.
\end{equation}
For equal tolerances $\delta_1=\delta_2=\delta$, the width of the shared acceptable range is
\begin{equation}\label{eq:shared_width}
    W=\max\bigl(0,\,2\delta-\Delta r\bigr).
\end{equation}
Hence, the compatibility boundary is $\Delta r^\ast=2\delta$, at which the shared acceptable range collapses to a single \ac{KPI} value.
\end{corollary}

Based on Definition~\ref{def:compatibility}, we propose a \emph{compatibility-aware conflict mitigation policy}. For each set $\mathcal{A}_p$ of agents in a direct conflict, the policy evaluates~\eqref{eq:pairwise_condition} for every pair $i,j\in\mathcal{A}_p$ and every \ac{KPI} $k\in\mathcal{K}_i\cap\mathcal{K}_j$. If the condition holds for every such pair and \ac{KPI}, the agents are compatible and are allowed to share control of $p$ without arbitration. If it fails for any pair on any shared \ac{KPI}, the agents are incompatible, and the policy prioritizes the higher-priority agent following~\cite{adamczyk2023conflict}.

%% file: sections/06_evaluation.tex
\section{Numerical Evaluation}\label{sec:evaluation}

In this section, we evaluate our compatibility-aware conflict mitigation policy through simulations of a direct conflict over the downlink \ac{PRB} allocation. We use the canonical two-agent, single-KPI case to isolate the effect of compatibility on arbitration and runtime behavior. We first describe the simulation and conflict scenario, then the metric and baselines, and finally the numerical results.

\subsection{Simulation Setup and Conflict Scenario}\label{subsec:setup}
To assess the compatibility of agents in AI-RAN and evaluate our conflict mitigation policy, we use NIST's \texttt{ns3-oran} module~\cite{garey2023oran}, an extension of \texttt{ns-3} that models O-RAN-like behavior and enables agents to monitor \acp{KPI} and control parameters of an LTE network.
Our simulations consider a single-cell topology to isolate the effects of direct conflicts from those of inter-cell interference, handovers, and inter-cell coordination.
We place stationary \acp{UE} according to the 3GPP TR~38.901~\cite{3gpp38901} UMa deployment model, so changes in the reported throughput result from changes to the \ac{PRB} allocation rather than mobility.
We use constant-rate UDP downlink traffic so throughput changes reflect \ac{PRB} control and cell capacity rather than transport-layer rate adaptation.
Each run lasts \unit[15]{s} with the \unit[2]{s} used for initialization, and \acp{KPI} are evaluated over the remaining \unit[13]{s}. 
We report median results across 100 independent runs with distinct random seeds, with shaded bands showing the interquartile range.
We summarize our simulation parameters in Table~\ref{tab:sim_setup}.

\begin{table}[t]
    \centering
    \renewcommand{\arraystretch}{1.12}
    \footnotesize
    \begin{tabular}{@{}p{0.25\linewidth}p{0.64\linewidth}@{}}
        \toprule
        \textbf{Component} & \textbf{Configuration} \\
        \midrule
        Simulator & \texttt{ns-3} with NIST's \texttt{ns3-oran} module~\cite{garey2023oran} \\
        Topology & Single-cell LTE/EPC; 1 eNB; stationary UEs uniformly placed in $[10,289]$~m \\
        Channel model & 3GPP TR~38.901 UMa path loss, probabilistic LOS/NLOS, log-normal shadowing~\cite{3gpp38901} \\
        Carrier frequency & 2110~MHz (LTE FDD band 1) \\
        Transmit power & eNB \unit[46]{dBm}~\cite{3gpp36104}; UE \unit[23]{dBm}~\cite{3gpp36101} \\
        MAC scheduler & Round-robin \\
        RRC & Ideal RRC model \\
        PRB control & Configured DL PRB allocation available to the scheduler, selected from $\{6, 15, 25, 50, 75, 100\}$~\cite{3gpp36104}; initial allocation of 50~PRBs \\
        Traffic & Downlink constant-rate UDP; 1400-byte packets \\
        Agents & TM and ES; $T_{\mathrm{TM}}\!=\!T_{\mathrm{ES}}\!=\!100$~ms; both monitor DL throughput and control the DL PRB allocation \\
        Simulation duration & \unit[15]{s} per run; \unit[2]{s} initialization; 100 runs per configuration \\
        \bottomrule
    \end{tabular}
    \vspace{-0.5em}
    \caption{Simulation parameters.}\label{tab:sim_setup}
    \vspace{-2.5em}
\end{table}

We consider a network infrastructure hosting two agents with distinct objectives commonly adopted in the conflict detection literature~\cite{giannopoulos2025comix,del2025pacifista}: a Throughput Maximization (TM) agent and an Energy Saving (ES) agent, both monitoring the aggregate downlink throughput \ac{KPI}, and controlling the downlink \ac{PRB} allocation parameter. 
The TM increases the \ac{PRB} allocation when the throughput falls below its acceptable range, whereas the ES reduces the allocation when its objective can be met with fewer radio resources~\cite{giannopoulos2025comix}. 
Both agents monitor throughput and take actions every \unit[100]{ms}, 
 i.e., $T_{TM}=T_{ES}=100$~ms, consistent with the near-real-time control loop from O-RAN~\cite{oran_wg3_ricarch}. 
Since the throughput delivered by our simulated cell depends on both the \ac{PRB} allocation and the traffic demand, we calibrate the agents' target based on the achievable throughput under each evaluated \ac{PRB} allocation to ensure they correspond to feasible network states.

\subsection{Evaluation Metric and Baselines}\label{subsec:metrics}

To quantify the performance impact 
of multiple agents sharing control of a parameter, we adopt the following metric:

\textbf{\ac{SLA} Drift:} 
We quantify the violation of each agent's objective using \ac{SLA} drift, adapting the normalized SLA-drift metric in~\cite{campos2025dreamin} to the two-sided deadband objectives considered in this work. For agent $a$, let $y(t_i)$ denote the reported throughput at sampling instant $t_i$, and let $l_a=r_a-\delta_a$ and $h_a=r_a+\delta_a$ denote its lower and upper deadband bounds. We define the instantaneous \ac{SLA} drift as
\begin{equation}\label{eq:drift}
    d_a(t_i)
    =
    \max\left(
    0,\,
    \frac{l_a-y(t_i)}{l_a},\,
    \frac{y(t_i)-h_a}{h_a}
    \right).
\end{equation}

Thus, $d_a(t_i)=0$ when the reported throughput lies within the agent's deadband. When the lower or upper bound is violated, $d_a(t_i)$ measures the relative distance to the corresponding deadband boundary. This extends the one-sided normalized SLA violation in~\cite{campos2025dreamin} to objectives with both lower and upper acceptable bounds.

We average the drift over the sampling instants in the evaluation interval:
\begin{equation}\label{eq:drift_pct}
    \mathrm{Drift}_a
    =
    \frac{100}{N_{\mathrm{eval}}}
    \sum_{i=1}^{N_{\mathrm{eval}}}
    d_a(t_i),
\end{equation}
where $N_{\mathrm{eval}}$ is the number of samples collected during the evaluation interval. The factor $100$ expresses the result as a percentage. All evaluated throughput bounds satisfy $l_a>0$, so the normalization in~\eqref{eq:drift} is well-defined.
To assess the impact of sharing control of a parameter on both the TM and ES agents, 
we report our results in terms of the \emph{combined SLA drift}, as an average of $\mathrm{Drift}_{\mathrm{TM}}$ and $\mathrm{Drift}_{\mathrm{ES}}$.

We compare our policy against two baselines representing its two operating choices: 
\textit{(i)} Ad Hoc Shared Control, an unarbitrated baseline in which every request is applied upon arrival, so the most recent request determines the \ac{PRB} allocation; and 
\textit{(ii)} Priority Mitigation, which assigns priorities to the agents. Following~\cite{adamczyk2023conflict}, decisions from the prioritized agent remain effective when conflicting requests target the same parameter, limiting lower-priority actions. Our policy permits shared control only after establishing compatibility; otherwise, it invokes Priority Mitigation. 
Without loss of generality, the ES has a higher priority than the TM in our simulations. 

\begin{figure}[t]
    \centering
    \includegraphics[width=\linewidth]{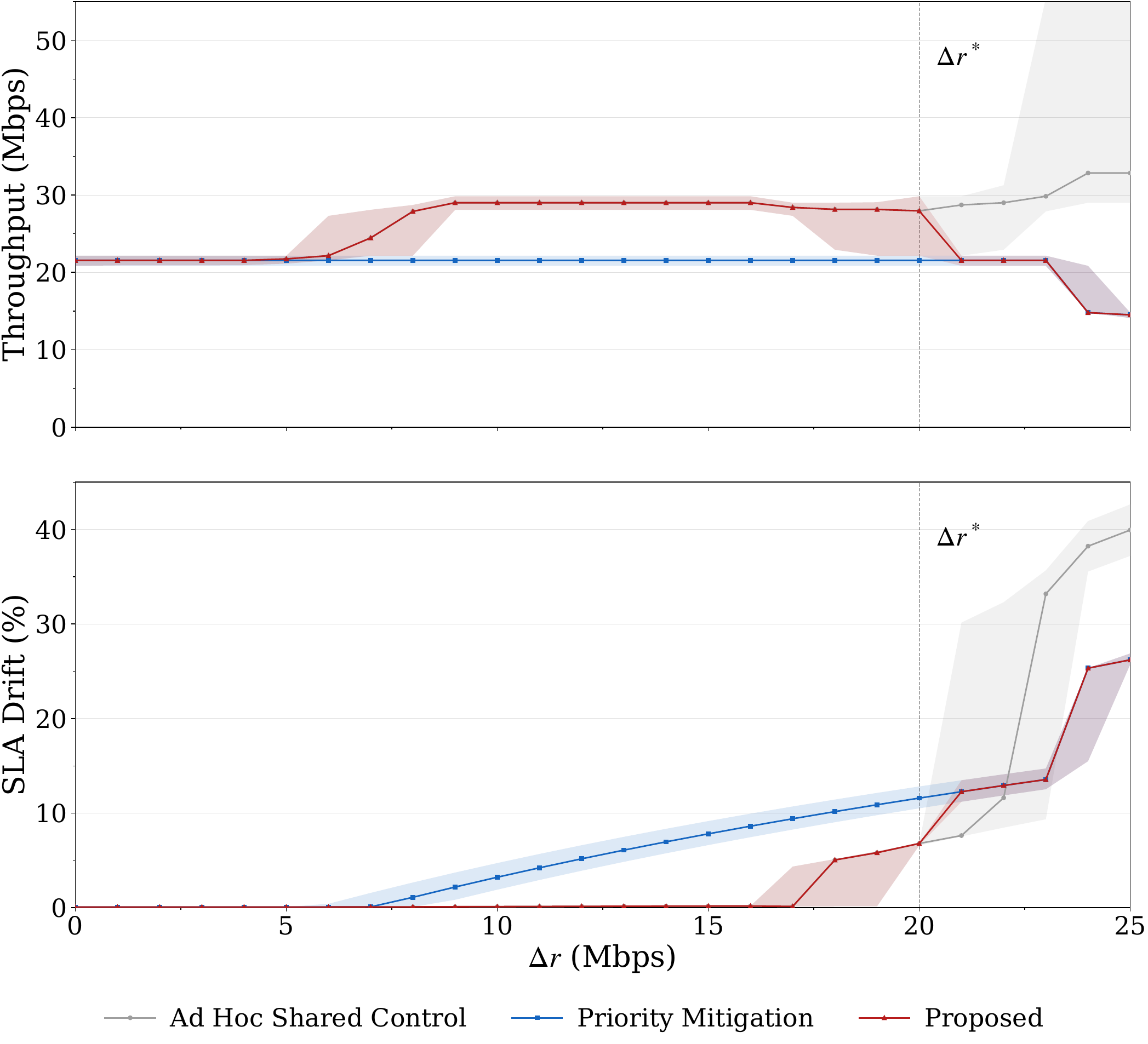}
    
    \caption{Impact of the coexistence between agents with different target separations, under our proposed policy and two baselines, in terms of combined \ac{SLA} drift and their monitored downlink throughput.  
    }
    \label{fig:sweep_deltar}
    \vspace{-.5cm}
\end{figure}

\subsection{Evaluating Compatibility Conditions}\label{subsec:design-eval}
In this analysis, we evaluate how differences between the agents' objectives impact 
their ability to coexist while sharing the control of a parameter.
Since their shared acceptable range is determined by their target values and tolerances, 
we set equal tolerances, $\delta_{\mathrm{TM}}=\delta_{\mathrm{ES}}=\unit[10]{Mbps}$, 
so that both agents have acceptable ranges of the same width, 
and vary the separation between their targets, $\Delta r$.
We consider 10 active \acp{UE} generating \unit[4.5]{Mbps} of traffic each,
and use \unit[28]{Mbps} (an achievable throughput under this traffic demand) as the calibration point for the agents' targets. 
We place the targets symmetrically around this reference value, with $r_{\mathrm{TM}}=28+\Delta r/2$ and $r_{\mathrm{ES}}=28-\Delta r/2$, yielding the shared acceptable throughput range $[\,18+\Delta r/2,\,38-\Delta r/2\,]$~Mbps. 
Varying $\Delta r$ allows us to reduce the alignment between the agents' objectives while keeping their targets centered around a feasible network state. 
By~\eqref{eq:shared_width}, its width $W$ decreases linearly with $\Delta r$ and collapses at $\Delta r^\ast=2\delta=\unit[20]{Mbps}$.

\Fig{sweep_deltar} shows the results of our simulations.
The Ad Hoc Shared Control maintains near-zero \ac{SLA} drift while the agents are compatible and the shared acceptable range is attainable with the evaluated \ac{PRB} allocations. As $\Delta r$ approaches $\Delta r^\ast$, the shared acceptable range narrows until the available \ac{PRB} allocations can no longer induce a throughput inside it, and the \ac{SLA} drift increases even though the agents remain compatible. In contrast, the Priority Mitigation suppresses one agent's actions at the cost of violating that agent's objective, even when the agents are compatible.

Our proposed policy displays three separate regimes according to $\Delta r$.
For $\Delta r\!\leq\!\unit[17]{Mbps}$, the \ac{SLA} drift remains below $0.2\%$, indicating that the agents' objectives are jointly satisfied. The throughput changes with $\Delta r$, but it remains within the shared acceptable range. 
For $\unit[18]{Mbps}\!\leq\!\Delta r\!\leq\!\unit[20]{Mbps}$, the agents remain compatible, but none of the discrete \ac{PRB} allocations induces a throughput within the narrowing shared acceptable range. This separates objective compatibility from practical attainability: the shared acceptable range remains non-empty, but no available \ac{PRB} configuration realizes it. Consequently, \ac{SLA} drift increases to ${\approx}5.8\%$ despite compatibility, with no substantial throughput degradation.
For $\Delta r\!>\!\Delta r^\ast$, the shared acceptable range is empty and agents are incompatible. 
Only then does our policy perform arbitration, prioritizing the ES. 
The combined \ac{SLA} drift grows with $\Delta r$, reaching ${\approx}26\%$ at $\Delta r\!=\!\unit[25]{Mbps}$, as TM's objective is no longer satisfied after its conflicting requests are suppressed.
These results show that our compatibility test can identify conditions where agents are compatible and can share control of parameters without unnecessary arbitration.

\begin{figure}[t]
    \centering
    \includegraphics[width=\linewidth]{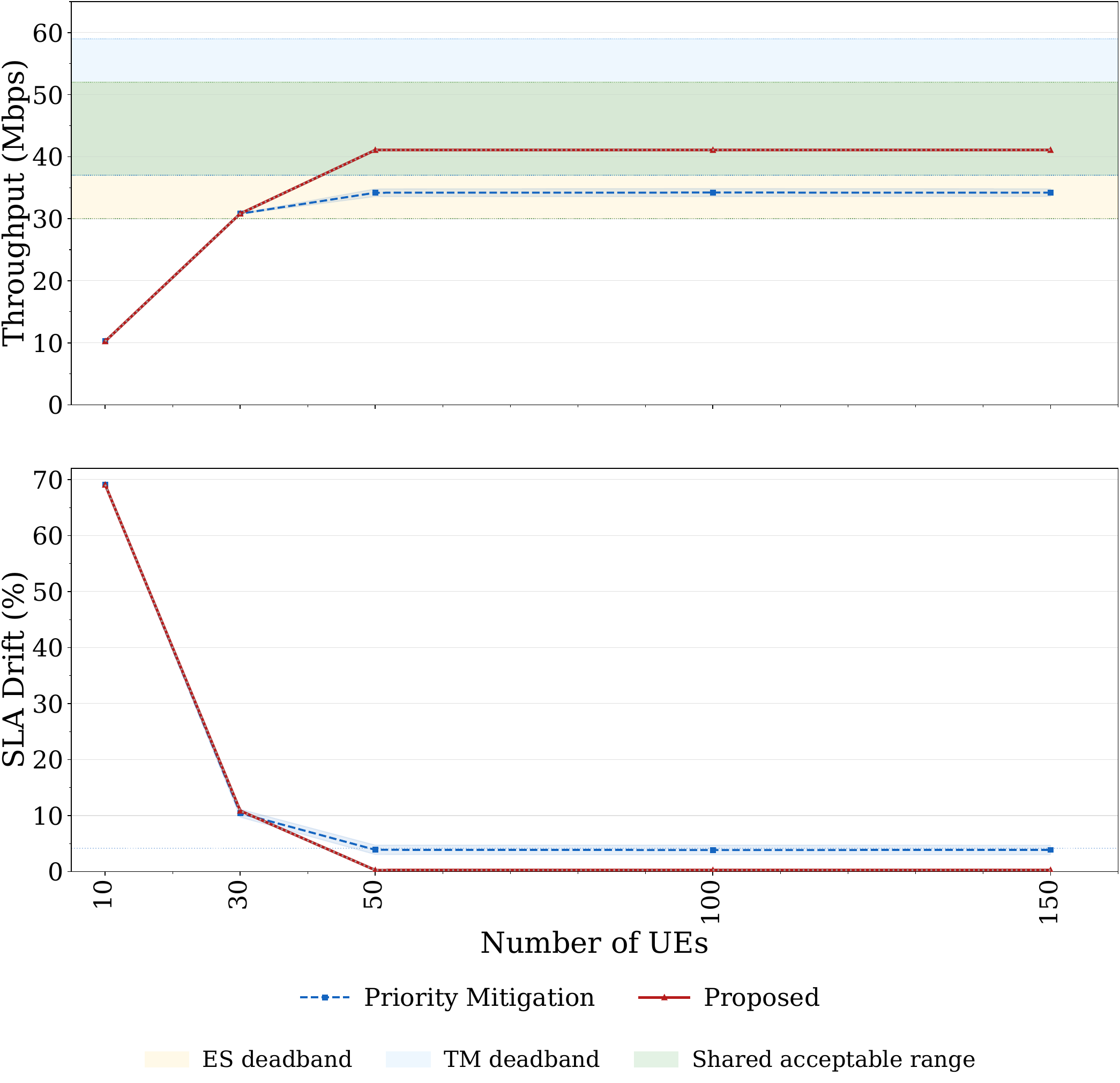}
    
    \caption{
    Effect of network conditions on the coexistence between agents,  as traffic demand increases with the number of active UEs, measured in terms of \ac{SLA} drift and throughput.}
    \label{fig:designs}
    \vspace{-1.5em}
\end{figure}

\subsection{From Compatibility Test to Practical Coexistence}\label{subsec:gain-eval}

In this analysis, we examine how network conditions affect the coexistence between agents. 
First, we consider an exemplar case where the TM and ES agents have equal tolerances of $\delta_{\mathrm{TM}}=\delta_{\mathrm{ES}}=\unit[11]{Mbps}$ and a target separation of \unit[7]{Mbps}, resulting in the shared acceptable throughput range $[37,52]$~Mbps.
To assess how the traffic demand affects the achievable throughput, we vary the number of active \acp{UE} from 10 to 150, generating \unit[1]{Mbps} of traffic each
as shown in \Fig{designs}. Here, the proposed policy permits shared control and is operationally equivalent to the Shared Control baseline, whose identical curve is omitted for readability. 
For up to 30 \acp{UE}, the traffic demand remains below the lower bound of the shared acceptable range.
Consequently, the \ac{RAN} cannot produce throughput within the shared acceptable range that satisfies both agents, despite their compatible objectives, and we can observe a high \ac{SLA} drift across all evaluated policies. 
Once the number of active \acp{UE} reaches 50, the cell can deliver a throughput within the shared acceptable range, and the mitigation policies begin to diverge. 
By permitting shared control, the proposed policy reduces the combined \ac{SLA} drift to approximately 0.2\%. 
In contrast, Priority Mitigation suppresses one agent and incurs a combined \ac{SLA} drift of ${\approx}4\%$, a $20\times$ increase. 
These results show that, while our compatibility test can identify whether the agents' objectives can be satisfied simultaneously, the network conditions and their impact on the monitored \acp{KPI} influence whether the agents can coexist in practice.

\begin{figure}[t]
    \centering
    \includegraphics[width=\linewidth]{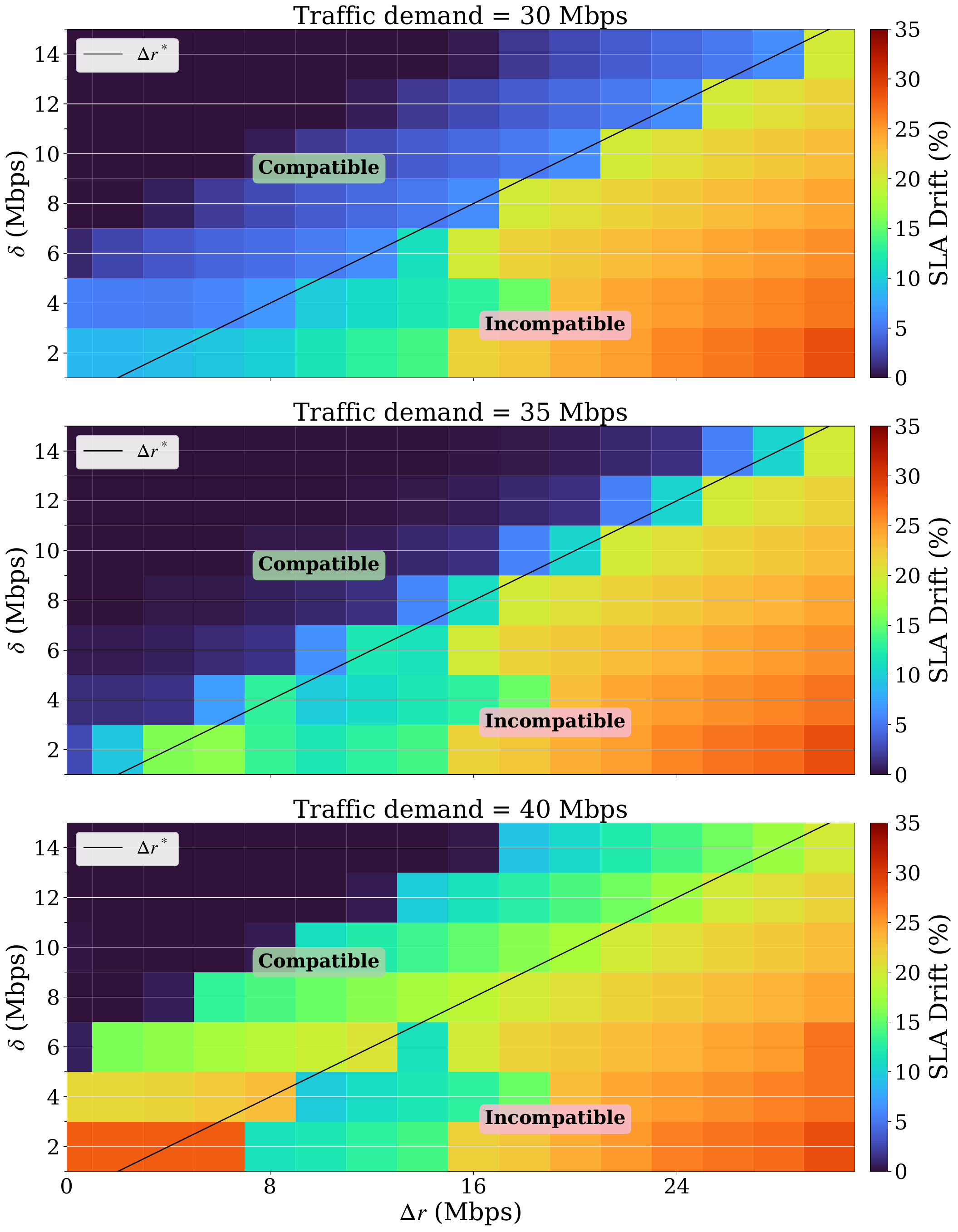}
    \vspace{-.75em}
    \caption{\ac{SLA} drift for different target separations and tolerances, for traffic demands of 30, 35, and \unit[40]{Mbps}. The line $\Delta r^\ast\!=\!2\delta$ delimits the compatible and incompatible regions.}
    \label{fig:friction_heatmap}
    \vspace{-1.5em}
\end{figure}

Finally, we extend the analysis across a wide range of target separations and tolerances under three traffic demands, and measure their effects on \ac{SLA} drift, as shown in \Fig{friction_heatmap}. The agents are designed around the \unit[35]{Mbps} operating point used in our reference configuration and the \unit[30]{Mbps} and \unit[40]{Mbps} cases evaluate how their behavior changes as the network operates below or above this reference.
The black line marks the theoretical compatibility boundary $\Delta r=2\delta$ from~\eqref{eq:shared_width}. 
When $\Delta r>2\delta$, the agents are incompatible because their acceptable ranges do not overlap. 
However, even when $\Delta r\leq2\delta$, we can observe some \ac{SLA} drift as the traffic demand moves away from the reference operating point.
This occurs because network conditions impose additional constraints on the network state, 
potentially making the shared acceptable range unreachable despite the agents being compatible. 
This shows that compatibility determines whether a shared acceptable range exists, while available control actions and network conditions determine whether it is attainable in practice.

%% file: sections/07_conclusion.tex
\section{Conclusion}\label{sec:conc}
In this paper, to the best of our knowledge, we provided the first control-theoretic formulation of objective compatibility among AI-RAN agents. Through this model, we derived compatibility conditions for determining whether agents can share control of the same parameters while allowing their objectives to be satisfied simultaneously, and proved that, under the componentwise deadband model, compatibility for any number of agents and monitored \acp{KPI} can be established through pairwise tests. We evaluated our model through simulations using NIST's \texttt{ns3-oran} module, and our results demonstrated that compatibility-aware mitigation can substantially reduce \ac{SLA} drift relative to priority mitigation when compatible objectives are attainable. Future work will relax the assumptions of the current formulation by incorporating an explicit RAN plant model to assess attainability alongside compatibility and to analyze the closed-loop behavior of incompatible agents under the deadband nonlinearity. We will also consider heterogeneous sampling periods, coupled multi-\ac{KPI} objectives, multiple shared control parameters, and indirect and implicit conflicts. Together, these directions would extend objective compatibility beyond the direct-conflict setting considered here toward more general coordination of heterogeneous agents in AI-RAN.